\documentclass[twoside]{aiml}

\usepackage{aimlmacro}

\usepackage{graphicx}
\usepackage{amsmath,amssymb}
\allowdisplaybreaks
\usepackage{xspace}
\usepackage[pdf,all]{xy}
\usepackage{tikz-cd}
\usepackage{enumitem}
\usepackage{booktabs,arydshln,tabularray,framed,dashbox}%
\usepackage{subcaption}
\usepackage{newtx}
\usepackage[T1]{fontenc}

\usepackage{algorithm,algorithmicx,algpseudocodex}
\algrenewcommand\algorithmicrequire{\textbf{Input:}}
\algrenewcommand\algorithmicensure{\textbf{Output:}}
\algnewcommand{\Initialize}[1]{\State\textbf{Initialize:} #1}
\algnewcommand\True{\textbf{true}\xspace}
\algnewcommand\False{\textbf{false}\xspace}
\usepackage{multicol}

\newcommand{\ab}{\ensuremath{A}\xspace}
\newcommand{\ag}{\text{\normalfont\textsf{Ag}}\xspace}

\newcommand{\CA}{\textnormal{\sf A}\xspace}
\newcommand{\CC}{\textnormal{\sf C}\xspace}
\newcommand{\CE}{\textnormal{\sf E}\xspace}
\newcommand{\CM}{\textnormal{\sf M}\xspace}
\newcommand{\CV}{\textnormal{\sf v}\xspace}
\newcommand{\CW}{\textnormal{\sf W}\xspace}

\newcommand{\lra}{\leftrightarrow}

\newcommand{\NEG}{\ensuremath{\mathord\sim}}
\renewcommand{\phi}{\varphi}
\newcommand{\pr}{\text{\normalfont\textsf{Prop}}\xspace}
\newcommand{\Ra}{\Rightarrow}
\newcommand{\ra}{\rightarrow}

\newcommand{\mbN}{\mathbb{N}}

\newcommand{\mcS}{\mathcal{S}}

\newcommand{\lang}{\ensuremath{\mathcal{EL}}\xspace}
\newcommand{\langc}{\ensuremath{\mathcal{ELC}}\xspace}
\newcommand{\langd}{\ensuremath{\mathcal{ELD}}\xspace}
\newcommand{\langm}{\ensuremath{\mathcal{ELF}}\xspace}
\newcommand{\langcd}{\ensuremath{\mathcal{ELCD}}\xspace}
\newcommand{\langcm}{\ensuremath{\mathcal{ELCF}}\xspace}
\newcommand{\langdm}{\ensuremath{\mathcal{ELDF}}\xspace}
\newcommand{\langcdm}{\ensuremath{\mathcal{ELCDF}}\xspace}
\newcommand{\langl}{\ensuremath{\mathcal{L}}\xspace}

\renewcommand{\l}{\text{\normalfont EL}\xspace}
\newcommand{\lc}{\text{\normalfont ELC}\xspace}
\newcommand{\ld}{\text{\normalfont ELD}\xspace}
\newcommand{\lm}{\text{\normalfont ELF}\xspace}
\newcommand{\lcd}{\text{\normalfont ELCD}\xspace}

\newcommand{\ldm}{\text{\normalfont ELDF}\xspace}
\newcommand{\lcdm}{\text{\normalfont ELCDF}\xspace}

\newcommand{\K}{\ensuremath{\mathbf{K}}\xspace}

\newcommand{\KB}{\ensuremath{\mathbf{EL}}\xspace}
\newcommand{\KBC}{\ensuremath{\mathbf{ELC}}\xspace}
\newcommand{\KBD}{\ensuremath{\mathbf{ELD}}\xspace}
\newcommand{\KBCD}{\ensuremath{\mathbf{ELCD}}\xspace}
\newcommand{\KBCDM}{\ensuremath{\mathbf{ELCDF}}\xspace}
\newcommand{\KBCM}{\ensuremath{\mathbf{ELCF}}\xspace}
\newcommand{\KBDM}{\ensuremath{\mathbf{ELDF}}\xspace}
\newcommand{\KBM}{\ensuremath{\mathbf{ELF}}\xspace}

\newcommand{\sysC}{\ensuremath{\mathbf{C}}\xspace}
\newcommand{\sysD}{\ensuremath{\mathbf{D}}\xspace}

\newcommand{\sysM}{\ensuremath{\mathbf{F}}\xspace}

\def\lastname{{\fontfamily{cmr}\fontsize{8pt}{8pt}\selectfont Liang and W\'{a}ng}}

\begin{document}

\begin{frontmatter}
  \title{Field Knowledge as a Dual to Distributed Knowledge}
  \footnote{This is a preprint of the paper published in Liao et al. (eds.) Fourth International Workshop on Logics for New-Generation Artificial Intelligence (LNGAI 2024), pp. 9--31, College Publications, 24 June 2024. \url{https://www.collegepublications.co.uk/LNGAI/?00004}}
  \subtitle{A characterization by weighted modal logic}
  \author{Xiaolong Liang}
  \address{School of Philosophy, Shanxi University\\92 Wucheng Road, Taiyuan, 030006, Shanxi, P.R. China}
  \author{Y\`{i} N. W\'{a}ng}
  \address{Department of Philosophy (Zhuhai), Sun Yat-sen University\\2 Daxue Road, Zhuhai, 519082, Guangdong, P.R. China}

  \begin{abstract}
The study of group knowledge concepts such as mutual, common, and distributed knowledge is well established within the discipline of epistemic logic. In this work, we incorporate epistemic abilities of agents to refine the formal definition of distributed knowledge and introduce a formal characterization of field knowledge. We propose that field knowledge serves as a dual to distributed knowledge. Our approach utilizes epistemic logics with various group knowledge constructs, interpreted through weighted models.
We delve into the eight logics that stem from these considerations, explore their relative expressivity and develop sound and complete axiomatic systems.\end{abstract}
\end{frontmatter}

\section{Introduction}

The introduction is segmented into two sections. In Section~\ref{subsec:notions}, we elucidate our interpretation of the concepts \emph{distributed knowledge} and \emph{field knowledge}. Section~\ref{subsec:approach} is dedicated to detailing our methodology for modeling these concepts within the context of weighted (or labeled) modal logic.

\subsection{Group notions of knowledge}
\label{subsec:notions}

Alice and Bob, both instrumentalists with additional expertise in philosophy and mathematics respectively, engage in a conversation that shapes and reflects their knowledge. Classical epistemic logic  \cite{Hintikka1962,FHMV1995,MvdH1995}  offers a framework for dissecting their individual and collective knowledge, utilizing tools like Kripke semantics among others.

Their mutual knowledge (a.k.a. everyone's knowledge or general knowledge) consists of statements known by both, essentially an intersection of their individual knowledge. In Kripke semantics, this is interpreted by the \emph{union} of their respective epistemic uncertainty relations. Common knowledge is recursive mutual knowledge: they know $\phi$, know that they know $\phi$, and so on, ad inf. It is modeled by the transitive closure of the union of their uncertainty relations.

Distributed knowledge signifies the sum of knowledge Alice and Bob would have after full communication, but it is not merely the union of what each knows. Though an interpretation based on the \emph{intersection} of their individual uncertainty relations does not fully align with its intended meaning either \cite{Roelofsen2007,AW2017rdk}, as it stands, this prevalent definition treats \emph{mutual knowledge as the semantic dual to distributed knowledge}.

In our scenario, we conceptualize distributed knowledge in light of the professional competencies of Alice and Bob. Their distributed knowledge of a statement $\phi$ is not understood as just a matter of aggregate knowledge, but rather the outcome of their combined expertise in musical instruments, philosophy, and mathematics. That is, $\phi$ is their distributed knowledge if their collaborative proficiency across these domains enables them to deduce $\phi$. Thus, we redefine \emph{distributed knowledge} as the \emph{union} of Alice and Bob's epistemic abilities, diverging from its classical interpretation.

Upon reevaluating distributed knowledge, we introduce the allied concept of \emph{field knowledge}. This notion encapsulates knowledge that stems from their shared discipline -- musical instruments, in this case. A statement $\phi$ falls under Alice and Bob's field knowledge if it is derivable exclusively from their musical background. The formal interpretation of this concept will be presented in Section~\ref{sec:semantics}, where we propose that \emph{field knowledge semantically functions as the dual to distributed knowledge}.

Developing a coherent characterization for the emergent concepts of distributed and field knowledge presents its challenges within classical epistemic logics. We aim to craft a comprehensive framework that encompasses these new ideas while preserving the established interpretations of mutual and common knowledge.

\subsection{Modeling knowledge in weighted modal logic}
\label{subsec:approach}

Even though the concept of \emph{similarity} is intrinsically linked to \emph{knowledge}, it has not been traditionally emphasized or explicitly incorporated in the classical representation of knowledge within the field of epistemic logic \cite{Hintikka1962,FHMV1995,MvdH1995}. Over recent years, researchers have started to probe this relationship more deeply, marking a fresh direction in the field \cite{NT2015,DLW2021}.  The technical framework for exploring this relationship has its roots in weighted modal logics \cite{LS1991,LM2014,HLMP2018}. This approach offers a quantitative way of considering similarity, allowing for a more nuanced understanding of knowledge.

In this paper, we adapt the concept of similarity from the field of data mining, where it is primarily used to quantify the likeness between two data objects. In data mining, distance and similarity measures are generally specific algorithms tailored to particular scenarios, such as computing the distance and similarity between matrices, texts, graphs, etc. (see, e.g., \cite[Chapter~3]{Aggarwal2015}). There is also a body of literature that outlines general properties of distance and similarity measures. For instance, in \cite{TSK2005}, it is suggested that typically, the properties of \emph{positivity} (i.e., $\forall x \forall y: s(x, y) = 1 \Ra x = y$) and \emph{symmetry} (i.e., $\forall x \forall y: s(x, y) = s(y, x)$) hold for $s(x, y)$ -- a binary numerical function that maps the similarity between points $x$ and $y$ to the range $[0, 1]$.

Our primary interest here does not lie in the measures of similarity themselves, but rather in modeling similarity and deriving from it the concepts of knowledge. Our work distinguishes itself from recent advancements in epistemic logic interpreted through the concepts of similarity or distance. One key difference is that we employ the standard language of epistemic logic. We do not explicitly factor in the degree of similarity into the language, maintaining the traditional structure of epistemic logic while reinterpreting its concepts in the light of similarity.

In our setting, the phrase ``$a$ knows $\phi$'' ($K_a\phi$) can be interpreted as ``$\phi$ holds true in all states that, in $a$'s perception by its expertise, resemble the actual state.'' A ``state'' in this context can be seen as a data object -- the focus of data mining. But it could also be treated as an epistemic object, a possible situation, and so forth. We generalize the similarity function by replacing its range [0, 1] with an arbitrary set of epistemic abilities. The degrees of similarity may not have a comparable or ordered relationship.

The primary focus of this paper is on \emph{group knowledge}, as elaborated in Section~\ref{subsec:notions}. We explore epistemic logics across all combinations of these group knowledge notions. As mutual knowledge is definable by individual knowledge (with only finitely many agents), we have formulated eight logics (with or without common, distributed, and field knowledge). The syntax and semantics of them are introduced in Sections~\ref{sec:syntax}--\ref{sec:models}, and in Section~\ref{sec:expressivity}, we compare the expressive power of these languages. 

For the axiomatization of the logics, we introduce sound and strongly complete axiomatic systems for the logics excluding common knowledge. For those incorporating common knowledge, we present sound and weakly complete axiomatic systems (owing to the lack of compactness for the common knowledge operators). These systems are then categorized based on whether their completeness results are obtainable via translation of models (Section~\ref{sec:completeness1}) or the canonical model method (Section~\ref{sec:completeness2}), shown via a path-based canonical model (Sections~\ref{sec:completeness3} and \ref{sec:completeness4}), or require a finitary method leading to a weak completeness result (Section~\ref{sec:completeness5}).

\section{Logics}

In this section, we present a comprehensive framework composed of eight distinctive logics. We supplement our discussion with illustrative examples, offering a visual representation of the models and their accompanying semantics.

\subsection{Syntax}\label{sec:syntax}\label{sec:semantics}

Our study utilizes formal languages rooted in the standard language of multi-agent epistemic logic \cite{FHMV1995,MvdH1995}, with the addition of modalities that represent group knowledge constructs. We particularly concentrate on the constructs of \emph{common knowledge}, \emph{distributed knowledge} and \emph{field knowledge}.

In terms of our assumptions, we consider \pr as a countably infinite set of propositional variables, and \ag as a finite, nonempty set of agents.

\begin{definition}(formal languages)
The languages utilized in our study are defined by the following rules, where the name of each language is indicated in parentheses on the left-hand side:
\vspace{-1.5ex}
\begin{align*}
&(\lang)&&\phi ::= p \mid \neg \phi \mid (\phi \ra \phi) \mid K_a\phi\\[-3.6pt]
&(\langc)&&\phi ::= p \mid \neg \phi \mid (\phi \ra \phi) \mid K_a\phi \mid C_G\phi\\[-3.6pt]
&(\langd)&&\phi ::= p \mid \neg \phi \mid (\phi \ra \phi) \mid K_a\phi \mid D_G\phi\\[-3.6pt]
&(\langm)&&\phi ::= p \mid \neg \phi \mid (\phi \ra \phi) \mid K_a\phi \mid F_G\phi\\[-3.6pt]
&(\langcd)&&\phi ::= p \mid \neg \phi \mid (\phi \ra \phi) \mid K_a\phi \mid C_G\phi \mid D_G\phi \\[-3.6pt]
&(\langcm)&&\phi ::= p \mid \neg \phi \mid (\phi \ra \phi) \mid K_a\phi \mid C_G\phi \mid F_G\phi \\[-3.6pt]
&(\langdm)&&\phi ::= p \mid \neg \phi \mid (\phi \ra \phi) \mid K_a\phi \mid D_G\phi \mid F_G\phi \\[-3.6pt]
&(\langcdm)&&\phi ::= p \mid \neg \phi \mid (\phi \ra \phi) \mid K_a\phi \mid C_G\phi \mid D_G\phi \mid F_G\phi\\[-2em]
\end{align*}
where $p \in \pr$, $a \in \ag$, and $G$ represents a nonempty subset of \ag, signifying a group. We also employ other boolean connectives, including conjunction ($\wedge$), disjunction ($\vee$), and equivalence ($\lra$). $E_G \phi$ is a shorthand for $\bigwedge_{a\in G} K_a\phi$ (note that $G$ is finite).
\end{definition}

We employ formulas such as $K_a \phi$ to depict: ``Agent $a$ knows $\phi$.'' This is often referred to as \emph{individual knowledge}. Similarly, formulas like $C_G\phi$, $D_G\phi$, $E_G\phi$ and $F_G\phi$ are used to convey that $\phi$ is \emph{common knowledge}, \emph{distributed knowledge}, \emph{mutual knowledge} (or \emph{everyone's knowledge}) and field knowledge of group $G$, respectively. When the group $G$ is a simple set, e.g., $\{a,b\}$, we write $C_{ab} \phi$ as a shorthand for $C_{\{a,b\}} \phi$, and likewise for the operators $D_{ab}$, $E_{ab}$ and $F_{ab}$.

Before delving into the formal semantics of these formulas, it is important to first establish the semantic models that will be used for the intended logics.

\subsection{Semantics}\label{sec:models}

We introduce a type of \emph{similarity models} for the interpretation of the languages.

\begin{definition}(similarity models)\label{def:models}
A \emph{similarity model} (\emph{model} for short) is a quintuple $(W,\ab,E,C,\nu)$ where:
\begin{itemize}[itemsep=0pt]
\item $W$ is a nonempty set of states or nodes, referred to as the \emph{domain};
\item \ab is an arbitrary set of abstract epistemic \emph{abilities} (e.g., one's expertise or profession), which could be empty, finite or infinite;%
\footnote{We have opted not to fix the set $A$ of abilities as a given parameter of the logic, in contrast to the set 
\ag of agents. The primary reason for this decision is our intention to examine models that may extend the set $A$ (see Section~\ref{sec:completeness1}). It is important to note that the validities and subsequent axiomatization of our logic remain unaffected when $A$ is considered to be an infinite set of abilities defined as a parameter of the logic.}
\item $E : W \times W \to \wp(\ab)$, known as an \emph{edge function}, assigns each pair of states a set of epistemic abilities, meaning that the two states are indistinguishable for individuals possessing only these epistemic abilities;
\item $C: \ag \to \wp(\ab)$ is a \emph{capability function} that assigns each agent a set of epistemic abilities;
\item $\nu: W \to \wp(\pr)$ is a valuation.
\end{itemize}
and conforms to the following conditions (for all $s, t \in W$):
\begin{itemize}[itemsep=0pt]
\item Positivity: if $E(s,t) = \ab$, then $s = t$;
\item Symmetry: $E(s,t) = E(t,s)$.
%\qed
\end{itemize}
\end{definition}

The above definition warrants further elucidation. Firstly, our approach adopts a broad interpretation of epistemic abilities that may not necessarily be arranged in a linear order, although such an arrangement is plausible \cite{DLW2021,LW2022}. Secondly, we perceive the edge function $E$ as a representation of the relation of similarity between states. The conditions of positivity and symmetry serve as generalized forms of common conditions employed to characterize similarity between data objects, as demonstrated in \cite{TSK2005}.%
\footnote{An implicit condition often assumed, the converse of positivity, posits $E(s,t)=A$ if $s=t$. This condition entails the reflexivity of graphs, depicted by the characterization axiom T (i.e., $K_a\phi \ra \phi$). In the realm of data mining, this condition implies that if two data objects are identical (i.e., they refer to the same data object), they would receive the maximum value from any similarity measure. This, however, is not always guaranteed.} Transitivity is usually not a characteristic of this framework (that $x$ and $y$, $y$ and $z$ are similar, does not necessarily mean that $x$ and $z$ are similar), resulting in the failure to uphold the principle of positive introspection ($K_a \phi \ra K_a K_a \phi$). Although it is easy to impose transitivity, we choose not to enforce it here. Our framework allows a more discerning evaluation of the tenability of positive introspection, and for an examination of other significant constraints, see \cite{LW2022}. Thirdly, in this context, similarities are deemed objective, signifying their constancy across diverse agents.

\begin{example}\label{ex1}
Alice and Bob are denoted as $a$ and $b$, respectively. Consider the fields mentioned in the beginning of the paper: musical instruments ($\alpha$), philosophy ($\beta$) and mathematics ($\gamma$), which are regarded as epistemic abilities in this example. As set up in the beginning of the paper, Alice is a philosopher, Bob a mathematician, and both of them are also instrumentalists. Four possible states are named $s_1, \dots, s_4$, in which $s_1$ is the factual state. From the viewpoint of an instrumentalist, all the states look no difference. From the perspective of a philosopher, no difference is between $s_1$ and $s_3$, and between $s_2$ and $s_4$. As for a mathematician, $s_1$ is indistinguishable to $s_2$, and $s_3$ to $s_4$. Consider the following propositions:

$p$: A standard modern piano has 88 keys in total. 

$q$: Knowledge is defined by ``justified true belief.''

$r$: Fermat's Last Theorem has been proved.

The above scenario can be abstracted to a pointed model $(M, s_1)$, such that $M = (W, A, E, C, \nu)$ and:
\begin{itemize}[itemsep=0pt]
	\item $W=\{s_1,s_2,s_3,s_4\}$; $\ab=\{\alpha,\beta,\gamma\}$;
	\item $E(s_1,s_2)=E(s_3,s_4)=\{\alpha,\gamma\}$, $E(s_1,s_3)=E(s_2,s_4)=\{\alpha,\beta\}$, $E(s_1,s_4)=E(s_2,s_3)=\{\alpha\}$, and for all $x \in W$, $E (x, x) = \{ \alpha, \beta, \gamma \}$;
	\item $C(a) = \{\alpha,\beta\}$ and $C(b) = \{\alpha,\gamma\}$;
	\item $\nu(s_1) = \{p,r\}$, $\nu(s_2) = \{p,q,r\}$, $\nu(s_3) = \{p\}$ and $\nu(s_4) = \{p,q\}$.
\end{itemize}

Figure~\ref{fig:model1} illustrates the model $M$ introduced above, where the factual state $s_1$ is framed by a rectangle and other states by an eclipse.
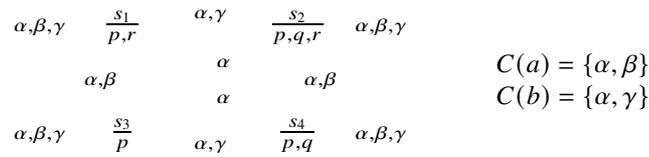
\begin{figure}[h]
\centering
\parbox{.6\textwidth}{%
\centering
$\xymatrix@R=2em@C=4em{
*++[F]{\frac{s_1}{p,r}}
\ar@{-}@(ul,dl)_{\alpha, \beta, \gamma}
\ar@{-}[r]^{\alpha,\gamma}
\ar@{-}[d]_{\alpha,\beta}
\ar@{-}@/^.3pc/[dr]^{\alpha}
&*++o[F]{\frac{s_2}{p,q,r}}
\ar@{-}@(ur,dr)^{\alpha, \beta, \gamma}
\ar@{-}[d]^{\alpha,\beta}
\ar@{-}@/^.3pc/[ld]^{\alpha}
\\
*++o[F]{\frac{s_3}{p}}
\ar@{-}@(ul,dl)_{\alpha, \beta, \gamma}
\ar@{-}[r]_{\alpha,\gamma}
&*++o[F]{\frac{s_4}{p,q}}
\ar@{-}@(ur,dr)^{\alpha, \beta, \gamma}
}$
}
\parbox{.23\textwidth}{%
$\begin{array}{l}
C(a) = \{\alpha, \beta\}\\
C(b) = \{\alpha, \gamma\}\\
\end{array}$
}
\caption{Illustration of the model in Example~\ref{ex1}.}\label{fig:model1}
\end{figure}

In the real world ($s_1$), one may come up with the following true sentences:
	\begin{itemize}[leftmargin=1em, itemsep=0pt]
	\item $K_a (p {\wedge} \neg q) {\wedge} \neg (K_a r {\vee} K_a \neg r)$ (Alice knows $p$ and $\neg q$, but doesn't know whether $r$.)
	\item $K_b (p {\wedge} r)\land \neg (K_a q {\vee} K_a \neg q)$ (Bob knows that $p$ and $r$, but doesn't know whether $q$.)
	\item $D_{ab}(p \wedge \neg q \wedge r)$ (It is Alice and Bob's distributed knowledge that $p$, not $q$, and $r$.)
	\item $F_{ab} p \land \neg (F_{ab} q \vee F_{ab} \neg q) \land \neg (F_{ab} r \vee F_{ab} \neg r)$ (While $p$ is Alice and Bob's field knowledge, $q$ and $r$ are not.)
	\end{itemize}
\end{example}

We now introduce a formal semantics that makes the model in Example~\ref{ex1} indeed yields the true sentences listed above.

\begin{definition}\label{def:semantics}
Given a formula $\phi$, a model $M = (W,\ab,E,C,\nu)$ and a state $s \in W$, we say $\phi$ is \emph{true} or \emph{satisfied} at $s$ in $M$, denoted $M,s \models \phi$, if the following hold (the case for $E_G\psi$ is redundant, but included for clarification):
$$\begin{array}{lll}
M,s \models p & \iff & p\in \nu(s)\\
M,s \models \neg\psi & \iff & \text{not } M,s \models \psi\\
M,s \models (\psi \ra \chi) & \iff & \text{if } M,s \models \psi \text{ then } M,s \models \chi\\
M,s \models K_a\psi & \iff & \text{for all $t\in W$, if $C(a)\subseteq E(s,t)$ then $M,t \models \psi$}\\
M,s \models E_G\psi & \iff & \text{$M,s \models K_a\psi$ for all $a \in G$}\\
M,s \models C_G\psi & \iff & \text{for all positive integers $n$, $M, s \models E_G^n \psi$}\\
M,s\models D_G\psi & \iff & \text{for all $t \in W$, if $\bigcup_{a \in G} C(a) \subseteq E(s,t)$ then $M,t \models \psi$}\\
M,s\models F_G\psi & \iff & \text{for all $t \in W$, if $\bigcap_{a \in G}C(a) \subseteq E(s,t)$ then $M,t \models \psi$,}\\
\end{array}$$
where $E_G^n\psi$ is defined recursively as $E^1_G E_G^{n-1}\psi$, with $E_G^1\psi$ to mean $E_G \psi$.
The concepts of \emph{validity} and \emph{satisfiability} have their classical meaning.
\end{definition}

In the definition above, the interpretation of $K_a \psi$ includes a condition ``$C(a) \subseteq E(s,t)$,'' which intuitively means that, ``Agent $a$, with its abilities, cannot discern between states $s$ and $t$.'' Thus, the formula $K_a \psi$ expresses that $\psi$ is true in all states $t$ that $a$ cannot differentiate from the current state $s$.

$E_G \psi$ stands for the conventional notion of \emph{everyone's knowledge}, or \emph{mutual knowledge} as we call, stating that ``Everyone in group $G$ knows $\psi$'' (see \cite{HM1992} for details).

\emph{Common knowledge} ($C_G\psi$) follows the classical fixed-point interpretation as $E_G C_G \psi$. In other words, $C_G\psi$ implies that, ``Everyone in group $G$ knows that $\psi$ is true, and everyone in $G$ knows about this first-order knowledge, and also knows about this second-order knowledge, and so on.''

The concept of \emph{distributed knowledge} ($D_G\psi$) in this paper diverges from the traditional definitions found in literature. We redefine distributed knowledge as being attainable by pooling together individual abilities. In practice, we swap the intersection of individual uncertainty relations with the union of individual epistemic abilities. Thus, $\psi$ is deemed distributed knowledge among group $G$ if and only if $\psi$ holds true in all states $t$ that, when utilizing all the epistemic abilities of agents in group $G$, cannot be differentiated from the present state.

An additional type of group knowledge, termed \emph{field knowledge} ($F_G\psi$), states that $\psi$ is field knowledge if and only if $\psi$ is true in all states $t$ that, using the shared abilities of group $G$, cannot be differentiated from the current state. We will examine its logical properties in greater detail later on.

Upon defining the semantics, we derive eight logics, each associated with one of the languages. We denote these logics using upright Roman capital letters. E.g., the logic corresponding to the interpretation of the language \langm is represented as \lm.

\begin{example}\label{ex2}
Consider the model illustrated in Figure~\ref{fig:model2}, we have the following:
\begin{enumerate}[itemsep=0pt]
	\item $M,s_2\models C_{ab}p\land E_{ab}p\land D_{ab}p\land \neg F_{ab} p$
	\quad (note that for $C_{ab}p$ we check whether $p$ is true in all states along an ``$ab$-path'' -- connected via $\{\lambda,\pi\}$, $\{\lambda,\mu\}$ or $\{\lambda,\pi,\mu\}$ edges -- that is, whether $p$ is true at $s_2$ through $s_4$, regardless of $s_1$.)
	\item $M,s_2\models F_{ab}q\land E_{ab}q\land D_{ab}q \land \neg C_{ab}q$
	\item $M,s_3\models D_{ab}r\land \neg K_a r\land \neg K_b r\land \neg E_{ab} r$
\end{enumerate}
\begin{figure}[h]
	\centering
	\parbox{.75\textwidth}{%
		\centering
		$\xymatrix@R=2em@C=4em{
			*++o[F]{\frac{s_1}{q}}
			\ar@{-}@(ul,ur)^{\lambda}
			\ar@{-}[r]_{\lambda}
			&*++o[F]{\frac{s_2}{p,q}}
			\ar@{-}@(ul,ur)^{\lambda, \pi, \mu}
			\ar@{-}[r]_{\lambda,\pi}
			&*++o[F]{\frac{s_3}{p,q,r}}
			\ar@{-}@(ul,ur)^{\lambda, \pi, \mu}
			\ar@{-}[r]_{\lambda,\mu}
			&*++o[F]{\frac{s_4}{p}}
			\ar@{-}@(ul,ur)^{\lambda, \pi}
		}$
	}
	\parbox{.23\textwidth}{%
		$\begin{array}{l}
			C(a) = \{\lambda, \pi\}\\
			C(b) = \{\lambda, \mu\}\\
		\end{array}$
	}
	\caption{Illustration of a model for Example~\ref{ex2}. We do not draw a line between two nodes when the edge between them is with no label (i.e., labeled by an empty set, e.g., between $s_1$ and $s_3$).}\label{fig:model2}
\end{figure}
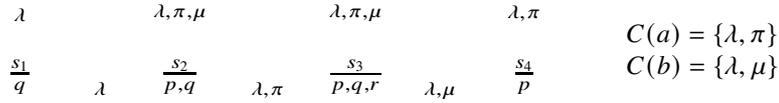
\end{example}

From the above, it is clear that none of  the following formula schemes are valid: $C_G\phi\ra F_G\phi$, $E_G\phi\ra F_G\phi$, $D_G\phi\ra F_G\phi$, $F_G\phi\ra C_G\phi$, $E_G\phi\ra C_G\phi$, $D_G\phi\ra C_G\phi$, $D_G\phi\ra E_G\phi$, $D_G\phi\ra K_a\phi$ (where $a\in G$). In particular, that $\phi$ is common knowledge implies that $\phi$ is distributed knowledge ($\models C_G\phi \ra D_G\phi$), but does not imply that it is field knowledge ($ \not\models C_G \phi \ra F_G \phi $). The underlying reasoning for this is that when professions intersect, the range of uncertain states can potentially expand significantly -- sometimes even more so than the increase that occurs when taking the transitive closure in the case of common knowledge. Consequently, this expansion of uncertainty can lead to a substantial contraction of field knowledge.
Nonetheless, the standard principles pertaining to individual, common, and distributed knowledge from classical logic remain applicable, as indicated by the following proposition.

\begin{proposition}
We have the following validities for any given formula $\phi$, any agent $a$ and any groups $G$ and $H$ (proofs omitted):
\vspace{-1ex}
\begin{multicols}{2}
\begin{enumerate}[itemsep=0pt]
\item $K_a(\phi \ra \psi) \ra (K_a\phi \ra K_a\psi)$
\item $\phi \ra K_a \neg K_a \neg \phi$
\item $C_{G}\phi \ra \bigwedge_{a\in G} K_a(\phi \wedge C_G \phi)$
\item $D_{\{a\}}\phi \lra K_a\phi$
\item $D_G \phi \ra D_H \phi$ (with $G \subseteq H$)
\item $\phi \ra D_G \neg D_G \neg \phi$
\item $F_{\{a\}}\phi \lra K_a\phi$
\item $F_G \phi \ra F_H \phi$ (with $H \subseteq G$)
\item $\phi \ra F_G \neg F_G \neg \phi$
\end{enumerate}
\end{multicols}
\end{proposition}

\subsection{Expressivity}
\label{sec:expressivity}

We adopt the conventional method for assessing a language's expressive power, which entails benchmarking it against the expressive capabilities of other languages. For an exact articulation of the relations in expressive power between two compatibly interpreted languages, we direct the reader to \cite[Def.~8.2]{vDvdHK2008}.

It is evident that the expressive power of all eight languages under consideration can be evaluated against one another. The comparative outcomes are encapsulated in Figure~\ref{fig:expressivity}, and the proofs are left in Appendix~\ref{sec:app-exp}.

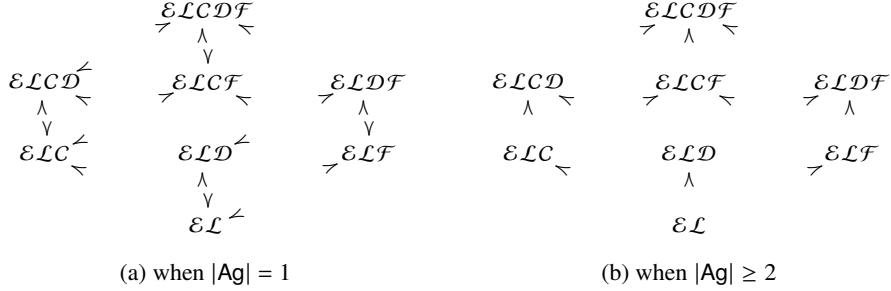
\begin{figure}[htbp]
\footnotesize
\subcaptionbox{when $|\ag|=1$}[.47\textwidth]{%
$\xymatrix@C=2.7em@R=1.6em{
&*++o[F]{\langcdm}\ar@<1pt>[dl]\ar@<1pt>[d]
\\
*++o[F]{\langcd}\ar@<1pt>[ur]\ar@<1pt>[d]&
*++o[F]{\langcm}\ar@<1pt>[u]\ar@<1pt>[dl]|(.5)\hole&
*++o[F]{\langdm}\ar[ul]\ar@<1pt>[dl]\ar@<1pt>[d]
\\
*++o[F]{\langc}\ar@<1pt>[u]\ar@<1pt>[ur]|(.5)\hole&
*++o[F]{\langd}\ar[ul]\ar@<1pt>[ur]\ar@<1pt>[d]&
*++o[F]{\langm}\ar[ul]|(.5)\hole\ar@<1pt>[u]\ar@<1pt>[dl]
\\
&*++o[F]{\lang}\ar[ul]\ar@<1pt>[u]\ar@<1pt>[ur]
}$
}
\hfill
\subcaptionbox{when $|\ag| \geq 2$}[.47\textwidth]{%
$\xymatrix@C=2.7em@R=1.6em{
&*++o[F]{\langcdm}
\\
*++o[F]{\langcd}\ar[ur]&*++o[F]{\langcm}\ar[u]&*++o[F]{\langdm}\ar[ul]
\\
*++o[F]{\langc}\ar[u]\ar[ur]|(.5)\hole&*++o[F]{\langd}\ar[ul]\ar[ur]&*++o[F]{\langm}\ar[ul]|(.5)\hole\ar[u]
\\
&*++o[F]{\lang}\ar[ul]\ar[u]\ar[ur]
}$
}
\caption{The above two diagrams illustrate the relative expressive power of the languages. An arrow pointing from one language to another implies that the second language is at least as expressive as the first. The ``at least as expressive as'' relationship is presumed to be \emph{reflexive} and \emph{transitive}, meaning that a language is considered at least as expressive as another if a path of arrows exists leading from the second to the first (self-loops exist for all, but omitted). A lack of a path of arrows from one language to another indicates that the first language is \emph{not} at least as expressive as the second. This implies that either the two languages are incomparable or that the first language is more expressive than the second.\label{fig:expressivity}}
\end{figure}

\section{Axiomatization}
\label{sec:ax}

We will present sound and complete axiomatic systems for the logics introduced in the preceding section. The names of these systems will be designated with bold capital letters. For example, the axiomatic system for the logic \lm is denoted as \KBM.

\subsection{Axiomatic systems}

The \K is a widely recognized axiomatic system for modal logic (here it refers to the multi-modal version with each $K_a$ functioning as a box operator). For simplicity, the axiom schemes are referred to as axioms in this context. The axiom system \textbf{KB} is obtained by augmenting the system \K with an additional axiom B (i.e., $\phi \ra K_a \neg K_a \neg \phi$). In this context, we represent \textbf{KB} as \K $\oplus$ B, where the symbol $\oplus$ acts like a union operation for sets of axioms and/or rules. For a comprehensive understanding of these axiomatic systems for modal logic, please refer to, say, \cite{BdRV2001}. Now, the system \KB that we introduce for our base logic \l is in fact \textbf{KB}.

Common knowledge is characterized by a set \sysC consisting of the following two inductive principles, represented by an axiom and a rule, which can be found in \cite{FHMV1995}:
\begin{itemize}[leftmargin=4em, itemsep=0pt]
\item[\textnormal{(C1)}] $C_{G}\phi \ra \bigwedge_{a\in G} K_a(\phi \wedge C_G \phi)$
\item[\textnormal{(C2)}] from $\phi \ra \bigwedge_{a\in G} K_a (\phi \wedge \psi)$ infer $\phi \ra C_G\psi$
\end{itemize}
Our the system \KBC is then represented as $\KB \oplus \sysC$.

Distributed knowledge is characterized by a set \sysD of additional axioms:
\begin{itemize}[leftmargin=4em, itemsep=0pt]
\item[\textnormal{(KD)}] $D_G(\phi \ra \psi) \ra (D_G\phi \ra D_G\psi)$
\item[\textnormal{(D1)}] $D_{\{a\}}\phi\lra K_a\phi$
\item[\textnormal{(D2)}] $D_G\phi\ra D_H\phi$ with $G\subseteq H$
\item[\textnormal{(BD)}] $\phi\ra D_G\neg D_G\neg\phi$
\end{itemize}
The resulting system for the logic \ld is then denoted as $\KBD = \KB \oplus \sysD$.

Field knowledge is characterized by the following set \sysM:
\begin{itemize}[leftmargin=4em, itemsep=0pt]
\item[\textnormal{(KF)}] $F_G(\phi \ra \psi) \ra (F_G\phi \ra F_G\psi)$
\item[\textnormal{(F1)}] $F_{\{a\}}\phi\lra K_a\phi$
\item[\textnormal{(F2)}] $F_G\phi\ra F_H\phi$ with $H\subseteq G$
\item[\textnormal{(BF)}] $\phi \ra F_G \neg F_G\neg\phi$
\item[(NF)] from $\phi$ infer $F_G \phi$
\end{itemize}
While the set \sysM might at first glance seem analogous to \sysD, there exists a subtle yet crucial difference between the axioms M2 and D2 -- specifically, the positions of the groups $G$ and $H$ are swapped. This distinction necessitates the introduction of the necessitation rule NF, while within the system \KBD, the rule ``from $\phi$ infer $D_G \phi$'' is derivable. The validity of these axioms can be verified with relative ease, and it is particularly noteworthy how the altered order of $G$ and $H$ precisely mirrors the union/intersection of epistemic abilities as observed in the semantic interpretations. Similarly, $\KBM$ is represented as $\KB \oplus \sysM$.

Moving towards more complex axiomatic systems, they are constructed in a similar manner. For any given string $\mathbf{\Xi}$ comprising elements from the set $\mathbf{\{C,D,M\}}$:
\begin{quote}
\itshape
The axiomatic system $\mathbf{EL \Xi}$ consists of all axioms and rules of \KB, along with those of the sets denoted by each character in string $\mathbf{\Xi}$.
\end{quote}
To illustrate, when $\mathbf{\Xi}$ is the string ``$\mathbf{CF}$,'' $\KBCM$ stands for $\KB \oplus \sysC \oplus \sysM$, and when $\mathbf{\Xi}$ is the string ``$\mathbf{CDF}$,'' then $\KBCDM$ equates to $\KB \oplus \sysC \oplus \sysD \oplus \sysM$. For an extreme case, when $\mathbf{\Xi}$ is an empty string, $\mathbf{EL \Xi}$ simply stands for $\KB$.

We now turn our attention to validating these axiomatic systems to be sound and complete for the corresponding logics. Soundness signifies that all the theorems of an axiomatic system are valid sentences of the corresponding logic. This can be simplified to the task of verifying that all the axioms of the system are valid, and that all the rules preserve this validity. The soundness of the proposed axiomatic systems can be confirmed without much difficulty. Though we omit the proof, we state it as the following theorem. We will follow this up with the completeness results in the subsequent section.

\begin{theorem}[soundness]
Every axiomatic system introduced in this section is sound for its corresponding logic.
\qed
\end{theorem}

\subsection{Completeness}\label{sec:completness}

In this section, we aim to demonstrate the completeness of all eight axiomatic systems that were introduced earlier.

It is a widely accepted fact in classical epistemic logic that the inclusion of common knowledge can cause a logic to lose its \emph{compactness}. This leads to the situation where its axiomatic system is not strongly complete, but only weakly complete (see, e.g., \cite{BdRV2001,vDvdHK2008}). This is also the case in our context. As a consequence, we will demonstrate that the four systems that do not include common knowledge are strongly complete axiomatic systems for their corresponding logics, while the other four systems that do incorporate common knowledge are only weakly complete.

The structure of this section is predicated on the various proof techniques we employ. We start with a method that reduces the satisfiability from classical epistemic logics to the logics we have proposed (Section~\ref{sec:completeness1}). However, this technique is only applicable to the system \KB; the canonical model method also works, and we provide a definition in Section~\ref{sec:completeness2} for the reference of the reader. When dealing with systems that incorporate either distributed or field knowledge, but not both (i.e., \KBD and \KBM), we utilize a path-based canonical model method (Section~\ref{sec:completeness3}). For the system that includes both distributed and field knowledge, namely \KBDM, while the path-based canonical model method still applies, a slightly more nuanced approach is required (Section~\ref{sec:completeness4}). Lastly, for the remaining systems incorporating common knowledge, we merge the finitary method (which involves constructing a closure) with the methods mentioned above (Section~\ref{sec:completeness5}).

\subsubsection{Proof by translation of satisfiability}
\label{sec:completeness1}

\begin{definition}(translation $\rho$)\label{def:trans-w}
The mapping $\cdot^\rho$ from symmetric Kripke models to similarity models is such that, for any symmetric Kripke model $N = (W,R,V)$, $N^\rho$ is the similarity model $(W,\ag \cup \{b\},E,C,\nu)$ with the same domain and:
\begin{itemize}[itemsep=2pt]
\item $E$ is such that for all $s,t \in W$, $E(s,t) = \{ a \in \ag \mid (s,t) \in R(a)\}$,
\item $b$ is a new agent which does not appear in \ag,
\item $C$ is such that for all $a \in \ag$, $C(a) = \{a\}$, and
\item $\nu$ is such that for all $s \in W$, $\nu(s) = \{ p \in \pr \mid s \in V(p)\}$.
\end{itemize}
\end{definition}
In the translated model $N^\rho$ of the above definition, the set of epistemic abilities is appointed as $\ag \cup \{b\}$. We use agents as labels of edges, which can intuitively be understood as an agent's inability to distinguish the ongoing state from current state when considering their epistemic abilities as a whole. In the subsequent lemma, we demonstrate that this translation preserves truth.

\begin{lemma}\label{lem:trans-w}\label{lem:sem-trans}
The following hold (proof in Appendix~\ref{sec:app-lem}):
\begin{enumerate}[itemsep=2pt]
\item\label{it:comp} Given a symmetric Kripke model $N$, its translation $N^\rho$ is a similarity model;
\item For any \langcd-formula $\phi$, any symmetric Kripke model $N$ and any state $s$ of $N$, $N,s \Vdash \phi$ iff $N^\rho,s \models \phi$.%
\footnote{The operator $F_G$ is undefined in classical epistemic logic, so there is no point to consider a language incorporated with this operator.}
\end{enumerate}
\end{lemma}

The system \KB is known to be complete for symmetric Kripke models. By Lemma \ref{lem:trans-w}, such a model can be translated to a similarity model in a way that preserves truth. So we get the following.

\begin{theorem}\label{thm:completeness2}
\KB is strongly complete for \l.\qed
\end{theorem}

We note that it is possible to use the same method to achieve completeness results for the systems \KBC, \KBD and \KBCD, as long as their completeness results in classical epistemic logic exist. However, to the best of our knowledge, the completeness of these systems in Kripke semantics, while expected, has never been explicitly established.

\subsubsection{Proof by the canonical model method}
\label{sec:completeness2}

We have demonstrated the completeness of \KB using the method of translation. This method is efficient and relies on the completeness result for its classical counterpart interpreted via Kripke semantics. In other words, the translation method cannot be employed for logics whose classical counterparts have not been introduced or studied (for example, logics with field knowledge), or when a completeness result does not exist for their classical counterparts (for instance, \lc, \ld and \lcd). In this section, we introduce a direct proof of the completeness of \KB using the canonical model method and extend this to completeness proofs for other logics in later sections.

\begin{definition}\label{def:cm-el}
The \emph{canonical model for \l} is the tuple $\CM = (\CW, \CA, \CE, \CC, \CV )$ such that:
\begin{itemize}[itemsep=0pt]
\item $\CW$ is the set of all maximal \KB-consistent sets of \lang-formulas;
\item $\CA=\wp(\ag)$, the power set of all agents;
\item $\CE : \CW \times \CW \to \wp(\CA)$ is defined such that for any $\Phi,\Psi \in \CW$, $E (\Phi,\Psi) = \bigcup_{a\in \ag} E_a(\Phi,\Psi)$, where $$E_a(\Phi,\Psi) = \left\{\begin{array}{ll}
	\CC(a), & \text{if $\{\chi\mid K_a\chi\in\Phi\}\subseteq\Psi$ and $\{\chi\mid K_a\chi\in\Psi\}\subseteq\Phi$},\\
	\emptyset, & \text{otherwise;}
\end{array}\right.$$
\item $\CC : \ag \to \wp(\CA)$ is such that for any agent $a$, $\CC(a)=\{G \subseteq \ag \mid a\in G \}$;
\item $\CV: \CW \to \wp(\pr)$ is such that for any $\Phi \in \CW$, $\CV (\Phi) = \{ p \in \pr \mid p \in \Phi \}$.
\end{itemize}
\end{definition}

The discerning reader may first ascertain that the canonical model for \l is indeed a model, and then proceed to demonstrate the completeness of \KB by employing conventional techniques. The specifics of these procedures are left in Appendix~\ref{sec:app-can}.

\subsubsection{Proof using a path-based canonical model}
\label{sec:completeness3}

When addressing logics that incorporate concepts of distributed and field knowledge, the conventional canonical model technique is not suitable. To overcome this challenge, there exists a methodology for logics with distributed knowledge \cite{FHV1992}. This technique, which traces its origins to the unraveling methods of \cite{Sahlqvist1975}, starts by analogously treating distributed knowledge as forms of individual knowledge. The process involves creating a pseudo model that embodies these aspects. This pseudo model is then unraveled into a tree-like structure with paths. The resulting structure is further processed through an identification/folding step to yield the target model.

An simplified approach has been proposed in \cite{WA2020}, where the construction of a path-based, tree-like model, referred to as a \emph{standard model}, is advocated. This approach eliminates the intermediate steps of unraveling and identification/folding. We embrace this latter approach here, setting out to construct a standard model directly.

\paragraph{\underline{Completeness of \KBD}}

\begin{definition}
$\langle \Phi_0, G_1, \Phi_1, \dots, G_n, \Phi_n \rangle$ is called a \emph{canonical path} for \ld, if:
\begin{itemize}[itemsep=0pt]
\item $\Phi_0, \Phi_1, \dots, \Phi_n$ represent maximal \KBD-consistent sets of \langd-formulas,
\item $G_1, \dots, G_n$ denote groups of agents, i.e., nonempty subsets of \ag.
\end{itemize}
\end{definition}
In the context of a canonical path $s = \langle \Phi_0,G_1,\Phi_1,\dots,G_n,\Phi_n \rangle$, we denote $\Phi_n$ as $tail(s)$. This also applies to canonical paths defined later.

\begin{definition}\label{def:sm-eld}
The \emph{standard model for \ld} is the tuple $\CM = (\CW, \CA, \CE, \CC, \CV)$ such that:
\begin{itemize}[itemsep=0pt]
\item \CW is the set of all canonical paths for \ld;
\item $\CA = \wp(\ag)$;
\item $\CE : \CW \times \CW \to \wp(\CA)$ is defined such that for any $s, t \in \CW$, $\CE(s,t)  = $
$$\left\{\begin{array}{ll}
	\bigcup_{a\in G} \CC(a), & \text{if $t$ is $s$ extended with $\langle G,\Psi \rangle$},\\[-3pt]
	&\text{$\{\chi\mid D_G\chi \in tail(s)\} \subseteq \Psi$ and $\{\chi\mid D_G\chi \in \Psi\} \subseteq tail(s)$};\\
	\bigcup_{a\in G} \CC(a), & \text{if $s$ is $t$ extended with $\langle G,\Psi \rangle$},\\[-3pt]
	&\text{$\{\chi\mid D_G\chi \in tail(t)\} \subseteq \Psi$ and $\{\chi\mid D_G\chi \in \Psi\} \subseteq tail(t)$};\\
	\emptyset, & \text{otherwise;}
\end{array}\right.$$
\item $\CC : \ag \to \wp(\CA)$ is such that for any agent $a$, $\CC(a)=\{G \in \ag \mid a \in G \}$;
\item $\CV: \CW \to \wp(\pr)$ is such that for all $s \in \CW$, $\CV (s) = \{p \in \pr \mid p \in tail(s)\}$.
\end{itemize}
\end{definition}

\begin{lemma}[standardness]
The standard model for \ld is a model.
\end{lemma}
\begin{proof}
Note that $\emptyset \notin \CC(a)$ for any agent $a$. This implies that for any $s,t \in \CW$, $\CE(s,t) \neq \CA$, thereby meeting the criterion of positivity. Additionally, the condition of symmetry is fulfilled as $\CE$ is a commutative function.
\end{proof}

\begin{lemma}[Truth Lemma]
Let $\CM = (\CW,\CA,\CE,\CC,\CV)$ be the standard model for \ld. For any $s \in \CW$ and \langd-formula $\phi$, $\phi \in tail(s)$ iff $\CM,s \models_{\ld} \phi$.
\end{lemma}
\begin{proof}
We prove it by induction of $\phi$. The boolean cases are easy by the definition of $\nu$ and the induction hypothesis. The only interested cases are $\phi=K_a\psi$ and $\phi=D_G\psi$.

Case $\phi=K_a\psi$: very similar to the case for $D_G \psi$ given below.

Case $\phi = D_G\psi$:
Suppose $D_G\psi\notin tail(s)$, but $\CM,s\models_\ld D_G\psi$. By definition, $\bigcup_{a\in G}\CC(a)\subseteq \CE (s,t)$ implies $\CM,t\models_\ld \psi$ for any $t\in \CW$. We can extend $\{\neg\psi\}\cup\{\chi\mid D_G\chi\in tail(s) \}\cup\{\neg D_G\neg\chi \mid \chi\in tail(s) \}$ to some maximal \KBD-consistent set $\Delta^+$. Similar to the proof of Lemma~\ref{lem:truth-KB} we get $\bigcup_{a\in G}\CC(a)\subseteq \CE(s,t)$ where $t$ extends $s$ with $\langle G,\Delta^+\rangle$. By the induction hypothesis we have $\CM,t\models_\ld \neg\psi$. A contradiction!

For the opposite direction, suppose $D_G\psi\in tail(s)$, but $\CM,s\not\models_\ld D_G\psi$. Then there exists $t\in W$ such that $\CM, t \not\models_\ld \psi$ and $\bigcup_{a\in G}C(a)\subseteq E(s,t)$. This implies that there exists $H \supseteq G$, such that $\{\chi \mid D_H \chi \in tail(s)\} \subseteq tail(t)$ and $\{\chi \mid D_H\chi \in tail(t) \} \subseteq tail(s)$. Since $D_G\psi\in tail(s)$ implies $D_H\psi\in tail(s)$, we have $\psi\in tail(t)$. By the induction hypothesis we have $\CM,t\models_\ld \psi$, also leading to a contradiction.
\end{proof}
\begin{theorem}
\label{thm:completeness3}
\KBD is strongly complete for \ld.
\qed
\end{theorem}

\paragraph{\underline{Completeness of \KBM}}

The completeness of \KBM can be demonstrated in a manner that parallels the completeness of \KBD. While we will not delve into the intricate details of the proofs, we will outline the necessary adaptations to the definitions of the standard model.

A canonical path for \lm mirrors that for \ld. The only modification required is the adjustment of the maximal consistent sets to align with the axiomatic system being considered.
When defining the standard model for \lm, we replace \CW with the set of all canonical paths for \lm, and let
$$\CE(s,t)  = \left\{\begin{array}{ll}
	\bigcap_{a\in G} \CC(a), & \text{if $t$ is $s$ extended with $\langle G,\Psi \rangle$,}\\[-3pt]
	&\text{$\{\chi\mid F_G\chi \in tail(s)\} \subseteq \Psi$ and $\{\chi\mid F_G\chi \in \Psi\} \subseteq tail(s)$},\\
	\bigcap_{a\in G} \CC(a), & \text{if $s$ is $t$ extended with $\langle G,\Psi \rangle$,}\\[-3pt]
	&\text{$\{\chi\mid F_G\chi \in tail(t)\} \subseteq \Psi$ and $\{\chi\mid F_G\chi \in \Psi\} \subseteq tail(t)$},\\
	\emptyset, & \text{otherwise.}
\end{array}\right.$$
Note that $\bigcup_{a\in G} \CC(a) = \{ H \mid H \cap G \neq \emptyset \}$, which includes all the shared epistemic abilities of those groups $H$ that intersects with $G$. Additionally, $\bigcap_{a\in G} \CC(a) = \{ H \mid G \subseteq H \}$, which represents all the supersets of $G$.

By using analogous proof structures, we can demonstrate the standardness of these models, show a Truth Lemma, and establish the completeness.

\begin{theorem}\label{thm:completeness4}
\KBM is strongly complete for \lm.
\qed
\end{theorem}

\subsubsection{Incorporation of both distributed and field knowledge}
\label{sec:completeness4}

We now discuss the logic and its axiomatic system that incorporate both distributed and field knowledge but exclude common knowledge, namely the logic \ldm and the axiomatic system \KBDM. The construction process requires careful consideration of the intricate interaction between the two types of knowledge modalities.

\begin{definition}\label{def:ELDF-path}
$\langle \Phi_0, I_1, \Phi_1, \dots, I_n, \Phi_n \rangle$ is a \emph{canonical path for \ldm}, if:
\begin{itemize}
\item $\Phi_0,\Phi_1,\dots,\Phi_n$ are maximal \KBDM-consistent sets of \langdm-formulas;
\item $I_1, \dots, I_n $ are of the form $(G, d)$ or $(G, m)$, with $G$ denoting a group, and ``$d$'' and ``$m$'' being just two distinct characters.
%\qed
\end{itemize}
\end{definition}

\begin{definition}
\label{def:sm-ldm}
The \emph{standard model for \ldm} is a tuple $\CM = (\CW, \CA, \CE, \CC, \CV)$ where \CA, \CC and \CV are defined just as in the standard model for \ld (Definition~\ref{def:sm-eld}), and:
\begin{itemize}
\item \CW is the set of all canonical paths for \ldm;
\item $\CE : \CW \times \CW \to \wp(\CA)$ is such that for any $s, t \in \CW$, $\CE(s,t)  = $
$$\left\{\begin{array}{ll}
	\bigcup_{a\in G} \CC(a), & \text{if $t$ extends $s$ with $\langle (G,d), \Psi \rangle$,}\\&\text{$\{\chi\mid D_G\chi \in tail(s)\} \subseteq \Psi$ and $\{\chi\mid D_G\chi \in \Psi \} \subseteq tail(s)$},\\
	\bigcup_{a\in G} \CC(a), & \text{if $s$ extends $t$ with $\langle (G,d), \Psi \rangle$,}\\&\text{$\{\chi\mid D_G\chi \in tail(t)\} \subseteq \Psi$ and $\{\chi\mid D_G\chi \in \Psi \} \subseteq tail(t)$},\\
	\bigcap_{a\in G} \CC(a), & \text{if $t$ extends $s$ with $\langle (G,m), \Psi \rangle$,}\\&\text{$\{\chi\mid F_G\chi \in tail(s)\} \subseteq \Psi$ and $\{\chi\mid F_G\chi \in \Psi \} \subseteq tail(s)$},\\
	\bigcap_{a\in G} \CC(a), & \text{if $s$ extends $t$ with $\langle (G,m), \Psi \rangle$,}\\&\text{$\{\chi\mid F_G\chi \in tail(t)\} \subseteq \Psi$ and $\{\chi\mid F_G\chi \in \Psi \} \subseteq tail(t)$},\\
	\emptyset, & \text{otherwise.}
\end{array}\right.$$
\end{itemize}
\end{definition}

The standardness and the Truth Lemma can be achieved in a similar way, and we leave a proof of the Truth Lemma in Appendix~\ref{sec:app-truth} for the careful reader.

\begin{theorem}
\label{thm:completeness5}
\KBDM is strongly complete for \ldm.
\qed
\end{theorem}

\subsubsection{Proof by a finitary standard model}
\label{sec:completeness5}

We now delineate the extension of the completeness results to the rest of the logics with common knowledge, deploying a finitary method for this purpose. We can only achieve weak completeness due to the non-compact nature of the common knowledge modality. A difficulty, except for \lc, is that we also need to address the modality for distributed or field knowledge.

In this section, we focus on providing the completeness proofs for \KBCDM. By making simple adaptations, one can obtain the completeness of the axiomatic systems for their sublogics with common knowledge. We adapt the definition of the \emph{closure} of a formula presented in \cite{WA2020}, to cater to formulas with modalities $D_G$ and/or $F_G$.

\begin{definition}\label{def:cl}
For an \langcdm-formula $\phi$, we define $cl(\phi)$ as the minimal set satisfying the subsequent conditions:
\begin{enumerate}[itemsep=2pt]
\item\label{it:cl-id} $\phi\in cl(\phi)$;
\item\label{it:cl-sub} if $\psi$ is in $cl(\phi)$, so are all subformulas of $\psi$;
\item\label{it:cl-neg} $\psi\in cl(\phi)$ implies ${\sim}\psi\in cl(\phi)$, where $\NEG\psi=\neg\psi$ if $\psi$ is not a negation and $\NEG\psi=\chi$ if $\psi=\neg\chi$;
\item\label{it:cl-1} $K_a\psi\in cl(\phi)$ implies $D_{\{a\}}\psi, F_{\{a\}}\psi\in cl(\phi)$;
\item\label{it:cl-d1} $D_{\{a\}}\psi\in cl(\phi)$ implies $K_a\psi\in cl(\phi)$;
\item\label{it:cl-d2} For groups $G$ and $H$, if $H$ appears in $\phi$, then $D_G\psi\in cl(\phi)$ implies $D_H\psi\in cl(\phi)$;
\item\label{it:cl-c1} $C_G\psi\in cl(\phi)$ implies $\{ K_a\psi, K_a C_G\psi \mid a \in G \} \subseteq cl(\phi)$;
\item\label{it:cl-m1} $F_G\psi\in cl(\phi)$ implies $\{ K_a\psi \mid a \in G\} \subseteq cl(\phi)$;
\item\label{it:cl-m2} For groups $G$ and $H$, if $H$ appears in $\phi$, then $F_G\psi\in cl(\phi)$ implies $F_H\psi\in cl(\phi)$.
\end{enumerate}
\end{definition}
Considering that the initial two clauses exclusively introduce subformulas of $\phi$, the subsequent three clauses incorporate formulas in a constrained manner, and given that there are a finite number of groups mentioned in $\phi$, with each group containing only a finite number of agents, we can readily confirm that $cl(\phi)$ for any given formula $\phi$.

Subsequently, we introduce the concept of a \emph{maximal consistent set of formulas within a closure}. For a comprehensive definition, which is naturally contingent on the specific axiomatic system under consideration, we refer to established literature, for example, \cite{vDvdHK2008}.

A \emph{canonical path for \lcdm in $cl(\phi)$} is defined similarly to that for \ldm (Def.~\ref{def:ELDF-path}). Given an \langcdm-formula $\phi$, we can construct the \emph{standard model for \lcdm with respect to $cl(\phi)$} in a manner that closely mirrors the construction of the standard model for \ldm (as per Def.~\ref{def:sm-ldm}). The primary differences lie in bounding the canonical paths by the closure and adjusting the logics accordingly. More specifically, we need to: (1) replace all occurrences of ``\langdm'' with ``\langcdm,'' and ``\ldm'' with ``\lcdm''; (2) within the definition of \CW, replace ``canonical paths for \ldm'' with ``canonical paths for \lcdm in $cl(\phi)$.'' Moreover, it is easy to confirm that the standard model for \lcdm (in any closure of a given formula) is a model.

\begin{lemma}[Truth Lemma]\label{lem:truthcdm}
Given an \langcdm-formula $\theta$, and let $\CM = (\CW, \CA, \CE, \CC, \CV)$ be the standard model for \lcdm with respect to $cl(\theta)$, for any $s \in \CW$ and $\phi \in cl(\theta)$, we have $\phi \in tail(s)$ iff $\CM,s \models_{\lcdm} \phi$.
\end{lemma}
\begin{proof}
We show the lemma by induction on $\phi$. We omit the straightforward cases (partially found in Appendix~\ref{sec:app-truth2}) and focus on the cases concerning common knowledge.

Suppose $C_G\psi\in tail(s)$, but $\CM,s\not\models_\lcdm C_G\psi$, then there are $s_i\in \CW$, $a_i\in G$, $0\leq i\leq n$ for some $n\in\mbN$ such that: $s_0=s$, $\CM,s_n\not\models\psi$ and $\CC(a_i)\subseteq \CE (s_{i-1},s_i)$ for $1\leq i\leq n$. Since $\CC(a_i)\subseteq \CE (s_{i-1},s_i)$, we have either $\{\chi\mid D_H\chi\in tail(s_{i-1})\}\subseteq tail(s_i)$ for some $H$ containing $a_i$ or $\{\chi\mid F_{\{a_i\}}\chi\in tail(s_{i-1})\}\subseteq tail(s_i)$. In both cases, $\{\chi\mid K_{a_i}\chi\in tail(s_{i-1})\}\subseteq tail(s_i)$. Since $C_G\psi\in tail(s_i)$ implies $K_{a_i}C_G\psi,K_{a_i}\psi\in tail(s_i)$, we can infer that $C_G\psi,\psi\in tail(s_n)$. By the induction hypothesis, $\CM,s_n\models_\lcdm \psi$, leading to a contradiction.

Suppose $C_G\psi\not\in tail(s)$, but $\CM,s\models_\lcdm C_G\psi$. Thus for any $s_i\in \CW$ and $a_i\in G$ where $0\leq i\leq n$, such that: $s_0=s$ and $\CC(a_i)\subseteq\CE(s_{i-1},s_i)$, we have $M,s_n\models_\lcdm \psi$ and $M,s_n\models_\lcdm C_G\psi$. Collect all such possible $s_n$ above and $s$ into the set $\mcS$; similarly collect all the $tail(s_n)$ and $tail(s)$ into the set $\Theta$. We define $\delta=\bigvee_{t\in\mcS}\widehat{tail(t)}$, where for any $t \in \CW$, $\widehat{tail(t)} = \bigwedge tail(t)$. (In general, for any finite set $\Psi$ of formulas, we write $\widehat{\Psi}$ for $\bigwedge \Psi$.)
We claim that $\vdash_{\lcdm}\delta\ra K_a\delta$ and $\vdash_{\lcdm}\delta\ra K_a\psi$ for any $a\in G$. By this claim and (C2) we have $\vdash_{\lcdm}\delta\ra C_G \psi$, and then by $\widehat{tail(s)}\ra\delta$ we have $\widehat{tail(s)} \ra C_G\psi$. In this way $C_G\psi\in tail(s)$, which leads to a contradiction. As for the proof of the claim:

(1) Suppose $\nvdash_{\lcdm}\delta\ra K_a\delta$, then $\delta\wedge \neg K_a\delta$ is consistent. Then there exists $t_0\in\mcS$ such that $\widehat{tail(t_0)}\wedge\neg K_a\delta$ is consistent. Notice that $\vdash_{\lcdm}\bigvee_{t\in \CW}\widehat{tail(t)}$, hence we have a consistent set $\widehat{tail(t_0)}\wedge\neg K_a\neg \widehat{tail(t_1)}$ for some $t_1\in \CW \setminus \mcS$; for otherwise we have $\CW\setminus\mcS=\emptyset$, hence $\vdash_{\lcdm}\delta$, which leads to $\vdash_{\lcdm}K_a\delta$, contradicting with $\nvdash_{\lcdm}\delta\ra K_a\delta$. Thus we have $\{\chi\mid K_a\chi\in tail(t_0)\}\subseteq tail(t_1)$, which implies $\{\chi\mid D_{\{a\}}\chi\in tail(t_0)\}\subseteq tail(t_1)$. Moreover, for any $\chi$, if $D_{\{a\}}\chi\in tail(t_1)$, then $\widehat{tail(t_0)}\wedge\neg K_a\neg D_{\{a\}}\chi$ is consistent, hence $\widehat{tail(t_0)}\wedge\chi$ is also consistent, thus $\chi\in tail(t_0)$. So we have $\{\chi\mid D_{\{a\}}\chi\in tail(t_1)\}\subseteq tail(t_0)$. Now we let $t_2$ be $t_0$ extended with $\langle(\{a\},d),tail(t_1)\rangle$, we have $\CC(a)\subseteq\CE(t_0,t_2)$. Hence $t_2\in\mcS$ but $tail(t_2)=tail(t_1)\notin\Theta$. A contradiction!

(2) Suppose $\nvdash_{\lcdm}\delta\ra K_a\psi$, then $\delta\wedge \neg K_a\psi$ is consistent. So there exists $t_0\in\mcS$ such that $\widehat{tail(t_0)}\wedge\neg K_a\psi$ is consistent. Thus $\{\NEG\psi\}\cup\{\chi\mid D_{\{a\}}\chi\in tail(t_0)\}\cup\{\neg D_{\{a\}}\NEG\chi \in cl(\theta)\mid \chi\in tail(t_0)\}$ is consistent. Hence it can be extended to some max consistent subset $\Delta^+$ in $cl(\theta)$. Let $t_1$ be $t_0$ extended with $\langle(\{a\},d),\Delta^+\rangle$, we have $\{\chi\mid D_{\{a\}}\chi\in tail(t_0)\}\subseteq tail(t_1)$. Moreover, if $D_{\{a\}}\chi\in tail(t_1) = \Delta^+$, then $\neg D_{\{a\}}\chi\notin\Delta^+$, thus $\neg D_{\{a\}} \NEG \neg \chi = \neg D_{\{a\}}\chi \in cl(\theta)$ and $\neg\chi \notin tail(t_0)$, which implies $\chi\in tail(t_0)$. So we also have $\{\chi\mid D_{\{a\}}\chi\in tail(t_1)\}\subseteq tail(t_0)$. Thus we have $\CC(a)\subseteq\CE(t_0,t_1)$. Hence $t_1\in\mcS$ and then $\CM,t_1\models_\lcdm\psi$, which contradicts with $\NEG\psi\in tail(t_1)$ by the induction hypothesis.
\end{proof}

\begin{theorem}
\KBCDM is weakly complete for \lcdm.
\qed
\end{theorem}

Building on the discussion earlier in this section, the completeness proofs for  \KBC, \KBCD and \KBCM can be derived from the proofs for  \KBCDM. For brevity, however, the intricate details of these adaptations are not included in this paper.

\section{Conclusion}

We examined epistemic logics with all combinations of common, distributed and field knowledge, interpreted in scenarios that consider agents' epistemic abilities, such as professions. We adopted a type of similarity model that extends from a Kripke model by adding weights to edges, and studied the axiomatization of the resulting logics.

The framework of our logics presents diverse possibilities for characterizing the concept of \emph{knowability}. Apart from interpreting knowability as known after a single announcement \cite{BBvDHHdL2008}, a group announcement \cite{ABDS2010}, or after a group resolves their knowledge \cite{AW2017rdk}, it is now conceivable to perceive knowability as known after an agent acquires certain skills (epistemic abilities) from some source or from a given group. Our framework also enables us to easily characterize \emph{forgetability} or \emph{degeneration} through changes in epistemic abilities, a process that is not as straightforward in classical epistemic logic.

Looking ahead, we aim to explore more sophisticated conditions on the similarity relation, such as those introduced in \cite{CMZ2009}. It would also be of interest to compare our framework with existing ones that use the same style of models, as presented in \cite{NT2015,DLW2021}.

\bibliographystyle{aiml}
\bibliography{main}

\newpage
\appendix
\section{Proofs regarding Experssivity}\label{sec:app-exp}

In Figure~\ref{fig:expressivity}, every language, with the exception of \langcdm, has an arrow pointing to its immediate superlanguages. This is clearly true, as by definition, every language is at most as expressive as its superlanguages. In the case when $|\ag|=1$, a reverse arrow also exists between languages that either both contain common knowledge or neither contain common knowledge.

\begin{lemma}\label{lem:exp1}
When $|\ag|= 1$ (i.e., when there is only one agent available in the language), 
\begin{enumerate}
\item $\lang\equiv\langd\equiv\langm\equiv\langdm$
\item $\langc\equiv\langcd\equiv\langcm\equiv\langcdm$
\item\label{it:exp-c} $\langdm \prec \langc$, and hence any language in the first clause are less expressive than any language in the second clause.
\end{enumerate}
\end{lemma}
\begin{proof}
1 \& 2. In the case when $|\ag|=1$ there is only one agent, and since $D_{\{a\}} \phi$ and $F_{\{a\}} \phi$ are equivalent to $K_a \phi$, the operators for distributed and mutual knowledge are redundant in this case. Hence the lemma.

3. We show that $\langc \not\preceq \langdm$, and so $\langdm \prec \langc$ since $\langdm\equiv\lang \preceq \langc$ by the first clause. Suppose towards a contradiction that there exists a formula $\phi$ of $\langdm$ equivalent to $C_{\{a\}}p$. Consider the set $\Phi=\{E_{\{a\}}^n p\mid n\in\mbN\}\cup\{\neg C_{\{a\}}p\}$. It is not hard to see that any finite subset of $\Phi$ is satisfiable, but not $\Phi$ itself. 
Let $k$ be the length of $\phi$ (refer to a modal logic textbook for its definition), and suppose $\{E_{\{a\}}^n p\mid n\in\mbN, n\leq k\}\cup\{\neg\phi\}$ is satisfied at a state $w$ in a model $M=(W,\ab,E,C,\nu)$. For any $s,t \in W$, we say that $s$ reaches $t$ in one step if $C(a)\subseteq E(s,t)$. Consider the model $M_k=(W_k,\ab,E,C,\nu)$, where $W_k$ is set of states reachable from $w$ in at most $k$ steps. We can verify that $M_k, w \models \{E_{\{a\}}^n p\mid n\in\mbN\}\cup\{\neg\phi\}$, which implies that $\Phi$ is satisfiable, leading to a contradiction.
\end{proof}

We now proceed to elucidate the absence of arrows in the figure for the case when $|\ag| \geq 2$.

\begin{lemma}\label{lem:exp2}
When $|\ag| \geq 2$ (i.e., when there are at least two agents available in the language), 
\begin{enumerate}
\item For any superlanguage $\langl$ of \langc, and any sublanguage $\langl'$ of $\langdm$, it is not the case that $\langl \preceq \langl'$;
\item For any superlanguage $\langl$ of \langd, and any sublanguage $\langl'$ of $\langcm$, it is not the case that $\langl \preceq \langl'$;
\item For any superlanguage $\langl$ of \langm, and any sublanguage $\langl'$ of $\langcd$, it is not the case that $\langl \preceq \langl'$.
\end{enumerate}
\end{lemma}
\begin{proof}
1. The Proof of Lemma~\ref{lem:exp1}(\ref{it:exp-c}) can be used here to show that $\langc \not\preceq \langdm$ also when $|\ag| \geq 2$.

2. Consider models $M=(W,\ab,E,C,\nu)$ and $M'=(W',\ab,E',C,\nu')$, where $\ab=\{1,2,3\}$, $C(a)=\{1,2\}$, $C(b)=\{1,3\}$ (if there are more agents in the language, they are irrelevant here), and are illustrated below.
\begin{center}
$M$\quad$\xymatrix@R=3em@C=4em{
			*+o[F]{\frac{u_1}{p}} \ar@{-}[d]_{1,2} \ar@{-}[r]^{1,3} & *+o[F]{\frac{u_2}{p}} \ar@{-}[d]^{1,2} \\
			*+o[F]{\frac{u_4}{p}} \ar@{-}[r]_{1,3} & *+o[F]{\frac{u_3}{p}}
		}$
\qquad\qquad
$M'$\quad$\xymatrix@R=3em@C=4em{
			*+o[F]{\frac{u'}{p}} \ar@{-}@(dl,dr)_{1,2,3} 
		}$
\vspace{2ex}
\end{center}
We can show by induction that for any formula $\phi$ of \langcm, $M,u_1\models\phi$ iff $M',u'\models\phi$. On the other hand, $M,u_1\models D_{ab}\bot$ but $M',u'\not\models D_{ab}\bot$.
It means that no \langcm-formula can discern between $M,u_1$ and $M',u'$, while languages with distributed knowledge can. Thus the lemma holds.

3. Consider similarity models $M=(W,\ab,E,C,\nu)$ and $M'=(W',\ab,E',C,\nu')$, where $\ab=\{1,2,3\}$, $C(a)=\{1,2\}$, $C(b)=\{1,3\}$, and are illustrated below.
\begin{center}
\vspace{2ex}
$M$\quad$\xymatrix@R=3em@C=4em{
			*+o[F]{\frac{u_1}{p}} \ar@{-}@(dl,dr)_{1,2,3} \ar@{-}[r]^{1} & *+o[F]{\frac{u_2}{}} \ar@{-}@(dl,dr)_{1,2,3} 
		}$
\qquad\qquad
$M'$\quad$\xymatrix@R=3em@C=4em{
			*+o[F]{\frac{u'}{p}} \ar@{-}@(dl,dr)_{1,2,3} 
		}$
\vspace{2ex}
\end{center}
We can show by induction that for any formula $\phi$ of \langcd, $M,u_1\models\phi$ iff $M',u'\models\phi$. Meanwhile, we have $M,u_1\not\models F_{ab}p$ and $M',u'\models F_{ab}p$. It follows that no \langcd-formula can discern between $M,u_1$ and $M',u'$, while languages with field knowledge can. Thus the lemma holds.
\end{proof}

As per Figure~\ref{fig:expressivity}, Lemma~\ref{lem:exp2} suggests that there is not an arrow or a path of arrows leading from \langc (or any language having an arrow or a path of arrows originating from \langc) to \langdm (or any language with an arrow or a path of arrows pointing to \langdm). Similar relationships exist between \langd and \langcm, and between \langm and \langcd. Furthermore, in Figure~\ref{fig:expressivity}, if there is an arrow or a path of arrows leading from one language to another, and not the other way round, this signifies that the first language is less expressive than the second. If there is no arrow or path of arrows in either direction between two languages, they are deemed incomparable. These observations lead us directly to the following corollary.

\begin{corollary}
When $|\ag|\geq 2$,
\begin{enumerate}
\item $\lang \prec \langc$, $\langd \prec \langcd$, $\langm \prec \langcm$ and $\langdm \prec \langcdm$;
\item $\lang \prec \langd$, $\langc \prec \langcd$, $\langm \prec \langdm$ and $\langcm \prec \langcdm$;
\item $\lang \prec \langm$, $\langc \prec \langcm$, $\langd \prec \langdm$ and $\langcd \prec \langcdm$;
\item \langc, \langd and \langm are pairwise incomparable;
\item \langcd, \langcm and \langdm are pairwise incomparable;
\item \langc is incomparable with \langdm;
\item \langd is incomparable with \langcm;
\item \langm is incomparable with \langcd.
\end{enumerate}
\end{corollary}

\section{Proof of Lemma~\ref{lem:trans-w}}\label{sec:app-lem}
We provide a proof for Lemma~\ref{lem:trans-w} while first repeating it:
\begin{lemma}
The following hold:
\begin{enumerate}
\item Given a symmetric Kripke model $N$, its translation $N^\rho$ is a similarity model;
\item For any \langcd-formula $\phi$, any symmetric Kripke model $N$ and any state $s$ of $N$, $N,s \Vdash \phi$ iff $N^\rho,s \models \phi$.
\end{enumerate}
\end{lemma}
\noindent\textbf{Proof.}
(i) Let $N = (W,R,V)$ be a symmetric Kripke model, and its translation $N^\rho = (W, \ag \cup \{b\}, E, C, \nu)$. For any $a \in \ag \cup \{b\}$ and $s, t\in W$, we have:
\[
\begin{array}{llll}
	a\in E(s,t) & \iff & (s,t)\in R(a) & \text{(Def.~\ref{def:trans-w})}\\
	& \iff & (t,s)\in R(a) & \text{(since $R(a)$ is symmetric)}\\
	& \iff & a\in E(t,s). & \text{(Def.~\ref{def:trans-w})}\\
\end{array}
\]
Hence $N^\rho$ satisfies symmetry. Furthermore, $N^\rho$ satisfies positivity, as there cannot be any $s,t \in W$ such that $E(s,t) = A \cup \{b\}$ and $s \neq t$. Hence $N^\rho$ is a similarity model.

(ii) Let $N = (W,R,V)$ and its translation $N^\rho = (W, \ag \cup \{b\}, E, C, \nu)$. We show the lemma by induction on $\phi$. The cases involving atomic propositions, Boolean connectives, and common knowledge are straightforward to verify since their semantic definitions follow a consistent pattern that facilitates smooth inductive reasoning. In this proof, we focus specifically on the cases for individual and distributed knowledge. It should be noted that the case for individual knowledge can be regarded as a particular instance of distributed knowledge; however, we include the details here for readers who seek a thorough clarity:
\[
\begin{array}{rcll}
N,s \Vdash K_a\psi & \iff & \text{for all $t \in W$, if $(s,t) \in R(a)$ then $N,t \Vdash \psi$}& \\
& \iff & \text{for all $t \in W$, if $a \in E(s,t)$ then $N,t \Vdash \psi$} & \\
& \iff &  \text{for all $t \in W$, if $C(a) \subseteq E(s,t)$ then $N,t \Vdash \psi$} & \\
& \iff & \text{for all $t \in W$, if $C(a) \subseteq E(s,t)$ then $N^\rho,t \models \psi$} &\\
& \iff & N^\rho,s \models K_a\psi. &\\[1em]
N,s \Vdash D_G\psi & \iff & \text{for all $t \in W$, if $(s,t) \in \bigcap_{a \in G} R(a)$, then $N,t \Vdash \psi$}& \\
& \iff & \text{for all $t\in W$, if $(s,t) \in R(a)$ for all $a\in G$, then $N,t \Vdash \psi$} & \\
& \iff & \text{for all $t\in W$, if $C(a)\subseteq E(s,t)$ for all $a\in G$, then $N,t \Vdash \psi$} & \\
& \iff & \text{for all $t\in W$, if $\bigcup_{a \in G} C(a)\subseteq E(s,t)$, then $N,t \models \psi$} &\\
& \iff & \text{for all $t\in W$, if $\bigcup_{a \in G} C(a)\subseteq E(s,t)$, then $N^\rho,t \models \psi$} &\\
& \iff & N^\rho,s \models D_G\psi. &\\
\end{array}
\]
\vspace{1em}

\section{Completeness of \KB by the Canonical Model Method}
\label{sec:app-can}

\begin{lemma}[canonicity]
The canonical model for \l is a similarity model.
\end{lemma}
\begin{proof}
Let $\CM=(\CW,\CA,\CE,\CC,\CV)$ be the canonical model for \l. Notice that $\emptyset \notin \CC(a)$ for any agent $a$, so $\CE(s,t)\neq\CA$ for any $s,t \in \CW$, ensuring positivity. The symmetry of the model is evident as $\CE(s,t)=\CE(t,s)$ for any $s,t \in \CW$. Therefore, $\CM$ is a similarity model.
\end{proof}

\begin{lemma}[Truth Lemma]
\label{lem:truth-KB}
Let $\CM = (\CW,\CA,\CE,\CC,\CV)$ be the canonical model for \l. For any $\Gamma \in \CW$ and any \lang-formula $\phi$, we have $\phi \in \Gamma$ iff $\CM,\Gamma \models_{\l} \phi$.
\end{lemma}
\begin{proof}
We will only demonstrate the case when $\phi$ is of the form $K_a \psi$ here.

Assuming $K_a\psi\in\Gamma$, but $\CM,\Gamma\not\models_\l K_a\psi$, there would exist a $\Delta\in \CW$ such that $\CC(a)\subseteq \CE (\Gamma,\Delta)$ and $\CM,\Delta\not\models_\l \psi$. Consequently, $\{\chi\mid K_a\chi\in\Gamma\}\subseteq\Delta$ (otherwise $\{a\} \notin \CE(\Gamma,\Delta)$, contradicting $\{a\} \in \CC(a)$). Thus, $\psi\in\Delta$. It follows from the induction hypothesis that $\CM,\Delta \models_\l \psi$, which results in a contradiction.

For the opposite direction, suppose $K_a\psi\notin \Gamma$, but $\CM,\Gamma\models_\l K_a\psi$, then for any $\Delta\in \CW$, $\CC(a)\subseteq \CE (\Gamma,\Delta)$ implies $\CM,\Delta\models_\l \psi$.
First, we assert that $\{\neg\psi\}\cup\{\chi\mid K_a\chi\in \Gamma\}\cup\{\neg K_a\neg\chi\mid \chi\in \Gamma\}$ is \KB consistent. If not, note that for any $\eta\in\{\chi\mid K_a\chi\in \Gamma\}$, we have $\neg K_a \neg K_a\eta\in \{\neg K_a\neg\chi\mid \chi\in \Gamma\}$. As $\vdash_\KB \neg K_a \neg K_a \eta \ra \eta$, it follows that $\{\neg\psi\}\cup\{\neg K_a\neg\chi\mid \chi\in \Gamma\}$ is not \KB consistent. Therefore, we have $\vdash_\KB \big(\bigwedge_{\chi\in\Gamma_0}\neg K_a\neg \chi \big) \ra \psi$ for some finite subset $\Gamma_0$ of $\Gamma$. This leads to $\vdash_\KB K_a\big((\bigwedge_{\chi\in\Gamma_0}\neg K_a\neg \chi) \ra \psi\big)$, and hence $\vdash_\KB \bigwedge_{\chi\in\Gamma_0}K_a\neg K_a\neg \chi \ra K_a\psi$. Since we have $\vdash_\KB \chi \ra K_a\neg K_a\neg \chi$ for any $\chi\in\Gamma_0$, it follows that we have $\vdash_\KB \big( \bigwedge_{\chi\in\Gamma_0}\chi \big) \ra K_a\psi$. This deduction implies that $K_a\psi \in \Gamma$, which contradicts our previous assumption.
Now, let us extend the set $\{\neg\psi\}\cup\{\chi\mid K_a\chi\in \Gamma\}\cup\{\neg K_a\neg\chi\mid \chi\in \Gamma\}$ to some maximal \KB-consistent set $\Delta^+$ of \lang-formulas. Notice that $K_a\chi\in \Gamma$ implies $\chi\in\Delta^+$ for any $\chi$. Furthermore, if we suppose $\chi\notin \Gamma$, then $\neg\chi\in \Gamma$, which leads to $\neg K_a\neg\neg\chi\in\Delta^+$, implying $\neg K_a\chi\in\Delta^+$. Therefore, $K_a\chi\in \Delta^+$ implies $\chi\in \Gamma$ for any $\chi$.
Given these stipulations, we find that $\CC(a)\subseteq\CE(\Gamma,\Delta)$. However, by using the induction hypothesis, we see that $M,\Delta\not\models_\l \psi$. As a result, $M,\Gamma\not\models_\l K_a\psi$. This conclusion contradicts our previous assumptions, confirming this direction of the lemma.
\end{proof}

With the Truth Lemma, we can state the following theorem:

\begin{theorem}[completeness of \KB, with a direct proof]
For any \lang-formula $\phi$ and any set $\Phi$ of \lang-formulas, if $\Phi \models_\l \phi$, then $\Phi \vdash_\KB \phi$.
\end{theorem}
\begin{proof}
To prove this, suppose the contrary: $\Phi \not\vdash_\KB \phi$. In this case, the set $\Phi\cup\{\neg\phi\}$ can be extended to a maximal \KB-consistent set $\Delta^+$. In the canonical model for \l, denoted \CM, we have $\CM,\Delta^+\models \chi$ for any formula $\chi\in \Phi\cup\{\neg\phi\}$. This conclusion leads to $\Phi \not\models_\l \phi$.
\end{proof}
\vspace{2em}

\section{Truth Lemma for Theorem~\ref{thm:completeness5}}
\label{sec:app-truth}

\begin{lemma}\label{lem:truthdm}
Let $\CM = (\CW, \CA, \CE, \CC, \CV)$ be the standard model for \ldm. For any $s \in \CW$ and any \langdm-formula $\phi$, $\phi\in tail(s)$ if and only if $\CM,s\models_{\ldm}\phi$.
\end{lemma}
\begin{proof}
The proof is by induction on $\phi$, and we only display the cases for modalities.

Case $\phi = K_a\psi$.
Suppose $K_a\psi\in tail(s)$, but $\CM,s\not\models_\ldm K_a\psi$, then there exists $t\in \CW$ such that $\CC(a)\subseteq \CE (s,t)$ and $\CM,t\not\models_\ldm \psi$. Therefore, $\{\chi\mid D_G \chi \in tail(s)\}\subseteq tail(t)$ for some group $G$ containing $a$ or $\{\chi\mid X_{\{a\}}\chi\in tail(s)\}\subseteq tail(t)$. In both scenarios, $\psi\in tail(t)$ since $K_a\psi\in tail(s)$ implies $D_G\psi, F_{\{a\}}\psi\in tail(s)$. By the induction hypothesis, we have $\CM,t\models_\ldm \psi$, which leads to a contradiction.
Suppose $K_a\psi\notin tail(s)$, but $\CM,s\models_\ldm K_a\psi$, then $\CC(a)\subseteq \CE (s,t)$ implies $\CM,t\models_\ldm \psi$ for any $t\in \CW$. Extend $\{\neg\psi\} \cup \{\chi\mid K_a\chi \in tail(s)\} \cup \{ \neg K_a \neg \chi \mid \chi \in tail(s) \}$ to some maximal \KBDM-consistent set $\Delta^+$, thus $\CC(a)\subseteq \CE(s,t)$ where $t$ extends $s$ with $\langle (\{a\}, d), \Delta^+ \rangle$. By the induction hypothesis, we have $\CM,t\models_\ldm \neg\psi$, which is contradictary.

Case $\phi=D_G\psi$.
Suppose $D_G\psi\in tail(s)$, but $\CM,s\not\models_\ldm D_G\psi$, then there exists some $t\in \CW$ such that $\bigcup_{a\in G}\CC(a)\subseteq \CE (s,t)$ and $\CM,t\not\models_\ldm \psi$. Therefore, $\{\chi\mid D_H\chi\in tail(s)\}\subseteq tail(t)$ for some group $H$ such that $G\subseteq H$. We have $\psi\in tail(t)$ since $D_G\psi\in tail(s)$ implies $D_H\psi\in tail(s)$. By the induction hypothesis, we have $\CM,t\models_\ldm \psi$, which leads to a contradiction.
Suppose $D_G\psi\notin tail(s)$, but $\CM,s\models_\ldm D_G\psi$, then $\bigcup_{a\in G}\CC(a)\subseteq \CE (s,t)$ implies $\CM,t\models_\ldm \psi$ for any $t\in \CW$. Extend $\{\neg\psi\} \cup \{\chi\mid D_G\chi \in tail(s) \cup \{ \neg D_G \neg \chi \mid \chi \in tail(s) \}\} $ to some maximal \KBDM-consistent set $\Delta^+$, thus $\bigcup_{a\in G}\CC(a)\subseteq \CE(s,t)$ where $t$ extends $s$ with $\langle (G,d),\Delta^+\rangle$. We have $\CM,t\models_\ldm \neg\psi$ by the induction hypothesis. A contradiction!

Case $\phi=F_G\psi$.
Suppose $F_G\psi\in tail(s)$, but $\CM,s\not\models_\ldm F_G\psi$, then there exists $t\in \CW$ such that $\bigcap_{a\in G}\CC(a)\subseteq \CE (s,t)$ and $\CM,t\not\models_\ldm \psi$. Therefore, $\{\chi\mid F_H\chi\in tail(s)\}\subseteq tail(t)$ for some group $H$ such that $H\subseteq G$ or $\{\chi\mid D_J\chi\in tail(s)\}\subseteq tail(t)$ for some group $J$ such that $G\cap J\neq\emptyset$. In both scenarios, we have $\psi\in tail(t)$ since $F_G\psi\in tail(s)$ implies $F_H\psi, D_J\psi\in tail(s)$. By the induction hypothesis, we have $\CM,t\models_\ldm \psi$, which leads to a contradiction.
Suppose $F_G\psi\notin tail(s)$, but $\CM,s\models_\ldm F_G\psi$, then $\bigcap_{a\in G}\CC(a)\subseteq \CE (s,t)$ implies $\CM,t\models_\ldm \psi$ for any $t\in \CW$. Extend $\{\neg\psi\} \cup \{\chi\mid F_G \chi \in tail(s) \} \cup \{ \neg F_G \neg \chi \mid \chi \in tail(s) \}$ to some maximal \KBDM-consistent set $\Delta^+$, thus $\bigcap_{a\in G}\CC(a)\subseteq \CE(s,t)$ where $t$ extends $s$ with $\langle (G,m),\Delta^+\rangle$. By the induction hypothesis, $\CM,t\models_\ldm \neg\psi$, which leads to a contradiction.
\end{proof}
\vspace{2em}

\section{Extra Cases for the Proof of Lemma~\ref{lem:truthcdm}}
\label{sec:app-truth2}

The proof of Lemma~\ref{lem:truthcdm} in the main text (p.~\pageref{lem:truthcdm}) only contains the case for common knowledge. Here we supplement the cases for distributed and field knowledge for the careful reader (individual knowledge can be treated as a special case of the two).

Case $\phi=D_G\psi$. The direction from $D_G\psi\in tail(s)$ to $\CM,s\models_\lcdm D_G\psi$ is similarly to the proof of Lemma \ref{lem:truthdm}. For the other direction, suppose $D_G\psi\not\in tail(s)$, but $\CM,s\models_\lcdm D_G\psi$, then $\bigcup_{a\in G}\CC(a)\subseteq \CE(s,t)$ implies $\CM,t\models_\lcdm \psi$ for any $t\in \CW$. Notice that $\{\NEG\psi\}\cup\{\chi\mid D_G\chi\in tail(t_0)\}\cup\{\neg D_G\NEG\chi \in cl(\theta)\mid \chi\in tail(t_0)\}$ is a consistent subset of $cl(\theta)$. Extend it to a maximal \KBCDM-consistent set $\Delta^+$ in $cl(\theta)$. Thus, by a similar method to the proof of $\vdash_\KBCDM \delta \to K_a\psi$ in Lemma~\ref{lem:truthcdm}, we have $\{\chi\mid D_G\chi\in tail(s)\}\subseteq \Delta^+$ and $\{\chi\mid D_G\chi\in \Delta^+\}\subseteq tail(s)$. Let $t$ be $s$ extended with $\langle(G,d),\Delta^+\rangle$, we have $\bigcup_{a\in G}\CC(a)\subseteq \CE(s,t)$. By the induction hypothesis we have $\CM,t\not\models_\lcdm \psi$, contradicting with $\CM,s\models_\lcdm D_G\psi$.

The case when $\phi=F_G\psi$ is similar to the case for distributed knowledge except that we extend the consistent set $\{\NEG\psi\}\cup\{\chi\mid F_G\chi\in tail(t_0)\}\cup\{\neg F_G\NEG\chi \in cl(\theta)\mid \chi\in tail(t_0)\}$ to get a maximal $\Delta^+$ in the closure, and let $t$ be $s$ extended with $\langle(G,m),\Delta^+\rangle$.

\end{document}